\documentclass[abstract]{scrartcl}

\usepackage{amsmath}
\usepackage{amssymb}
\usepackage{amsthm}
\usepackage[english]{babel}
\usepackage{caption}
\usepackage[inline]{enumitem}
\usepackage[utf8]{inputenc}
\usepackage[english, iso]{isodate}
\usepackage{mathtools}
\usepackage{microtype}
\usepackage{stmaryrd}
\usepackage{thmtools}
\usepackage{tikz}

\usepackage[unicode]{hyperref}
\usepackage[capitalise]{cleveref}

\DeclareMathSymbol{:}{\mathpunct}{operators}{"3A}

\DeclareUnicodeCharacter{00AC}{\neg}
\DeclareUnicodeCharacter{00B1}{\pm}
\DeclareUnicodeCharacter{00D7}{\times}
\DeclareUnicodeCharacter{00F7}{\div}
\DeclareUnicodeCharacter{0391}{A}
\DeclareUnicodeCharacter{0392}{B}
\DeclareUnicodeCharacter{0393}{\varGamma}
\DeclareUnicodeCharacter{0394}{\varDelta}
\DeclareUnicodeCharacter{0395}{E}
\DeclareUnicodeCharacter{0396}{Z}
\DeclareUnicodeCharacter{0397}{H}
\DeclareUnicodeCharacter{0398}{\varTheta}
\DeclareUnicodeCharacter{0399}{I}
\DeclareUnicodeCharacter{039A}{K}
\DeclareUnicodeCharacter{039B}{\varLambda}
\DeclareUnicodeCharacter{039C}{M}
\DeclareUnicodeCharacter{039D}{N}
\DeclareUnicodeCharacter{039E}{\varXi}
\DeclareUnicodeCharacter{039F}{O}
\DeclareUnicodeCharacter{03A0}{\varPi}
\DeclareUnicodeCharacter{03A1}{P}
\DeclareUnicodeCharacter{03A3}{\varSigma}
\DeclareUnicodeCharacter{03A4}{T}
\DeclareUnicodeCharacter{03A5}{\varUpsilon}
\DeclareUnicodeCharacter{03A6}{\varPhi}
\DeclareUnicodeCharacter{03A7}{X}
\DeclareUnicodeCharacter{03A8}{\varPsi}
\DeclareUnicodeCharacter{03A9}{\varOmega}
\DeclareUnicodeCharacter{03B1}{\alpha}
\DeclareUnicodeCharacter{03B2}{\beta}
\DeclareUnicodeCharacter{03B3}{\gamma}
\DeclareUnicodeCharacter{03B4}{\delta}
\DeclareUnicodeCharacter{03B5}{\varepsilon}
\DeclareUnicodeCharacter{03B6}{\zeta}
\DeclareUnicodeCharacter{03B7}{\eta}
\DeclareUnicodeCharacter{03B8}{\theta}
\DeclareUnicodeCharacter{03B9}{\iota}
\DeclareUnicodeCharacter{03BA}{\kappa}
\DeclareUnicodeCharacter{03BB}{\lambda}
\DeclareUnicodeCharacter{03BC}{\mu}
\DeclareUnicodeCharacter{03BD}{\nu}
\DeclareUnicodeCharacter{03BE}{\xi}
\DeclareUnicodeCharacter{03BF}{o}
\DeclareUnicodeCharacter{03C0}{\pi}
\DeclareUnicodeCharacter{03C1}{\rho}
\DeclareUnicodeCharacter{03C3}{\sigma}
\DeclareUnicodeCharacter{03C2}{\varsigma}
\DeclareUnicodeCharacter{03C4}{\tau}
\DeclareUnicodeCharacter{03C5}{\upsilon}
\DeclareUnicodeCharacter{03C6}{\varphi}
\DeclareUnicodeCharacter{03C7}{\chi}
\DeclareUnicodeCharacter{03C8}{\psi}
\DeclareUnicodeCharacter{03C9}{\omega}
\DeclareUnicodeCharacter{03D1}{\vartheta}
\DeclareUnicodeCharacter{03D5}{\phi}
\DeclareUnicodeCharacter{03D6}{\varpi}
\DeclareUnicodeCharacter{03F1}{\varrho}
\DeclareUnicodeCharacter{03F4}{\Theta}
\DeclareUnicodeCharacter{03F5}{\epsilon}
\DeclareUnicodeCharacter{01F4}{\'{G}}
\DeclareUnicodeCharacter{1E30}{\'{K}}
\DeclareUnicodeCharacter{1E3E}{\'{M}}
\DeclareUnicodeCharacter{1E54}{\'{P}}
\DeclareUnicodeCharacter{1E82}{\'{W}}
\DeclareUnicodeCharacter{01F5}{\'{g}}
\DeclareUnicodeCharacter{1E31}{\'{k}}
\DeclareUnicodeCharacter{1E3F}{\'{m}}
\DeclareUnicodeCharacter{1E55}{\'{p}}
\DeclareUnicodeCharacter{1E83}{\'{w}}
\DeclareUnicodeCharacter{2026}{\ldots}
\DeclareUnicodeCharacter{2102}{\mathbb{C}}
\DeclareUnicodeCharacter{210D}{\mathbb{H}}
\DeclareUnicodeCharacter{2115}{\mathbb{N}}
\DeclareUnicodeCharacter{2119}{\mathbb{P}}
\DeclareUnicodeCharacter{211A}{\mathbb{Q}}
\DeclareUnicodeCharacter{211D}{\mathbb{R}}
\DeclareUnicodeCharacter{2124}{\mathbb{Z}}
\DeclareUnicodeCharacter{2190}{\leftarrow}
\DeclareUnicodeCharacter{2192}{\rightarrow}
\DeclareUnicodeCharacter{21A6}{\mapsto}
\DeclareUnicodeCharacter{21D0}{\Leftarrow}
\DeclareUnicodeCharacter{21D2}{\Rightarrow}
\DeclareUnicodeCharacter{21DD}{\rightsquigarrow}
\DeclareUnicodeCharacter{2200}{\forall}
\DeclareUnicodeCharacter{2202}{\partial}
\DeclareUnicodeCharacter{2203}{\exists}
\DeclareUnicodeCharacter{2204}{\nexists}
\DeclareUnicodeCharacter{2205}{\varnothing}
\DeclareUnicodeCharacter{2208}{\in}
\DeclareUnicodeCharacter{2209}{\notin}
\DeclareUnicodeCharacter{220B}{\ni}
\DeclareUnicodeCharacter{220C}{\not\ni}
\DeclareUnicodeCharacter{220F}{\prod}
\DeclareUnicodeCharacter{2210}{\coprod}
\DeclareUnicodeCharacter{2211}{\sum}
\DeclareUnicodeCharacter{2213}{\mp}
\DeclareUnicodeCharacter{2215}{/}
\DeclareUnicodeCharacter{2216}{\setminus}
\DeclareUnicodeCharacter{2217}{\ast}
\DeclareUnicodeCharacter{2218}{\circ}
\DeclareUnicodeCharacter{2219}{\bullet}
\DeclareUnicodeCharacter{221E}{\infty}
\DeclareUnicodeCharacter{2223}{\mid}
\DeclareUnicodeCharacter{2224}{\nmid}
\DeclareUnicodeCharacter{2225}{\parallel}
\DeclareUnicodeCharacter{2226}{\nparallel}
\DeclareUnicodeCharacter{2227}{\wedge}
\DeclareUnicodeCharacter{2228}{\vee}
\DeclareUnicodeCharacter{2229}{\cap}
\DeclareUnicodeCharacter{222A}{\cup}
\DeclareUnicodeCharacter{2260}{\ne}
\DeclareUnicodeCharacter{2261}{\equiv}
\DeclareUnicodeCharacter{2262}{\not\equiv}
\DeclareUnicodeCharacter{2264}{\leq}
\DeclareUnicodeCharacter{2265}{\geq}
\DeclareUnicodeCharacter{2248}{\approx}
\DeclareUnicodeCharacter{2282}{\subset}
\DeclareUnicodeCharacter{2283}{\supset}
\DeclareUnicodeCharacter{2284}{\not\subset}
\DeclareUnicodeCharacter{2285}{\not\supset}
\DeclareUnicodeCharacter{2286}{\subseteq}
\DeclareUnicodeCharacter{2287}{\supseteq}
\DeclareUnicodeCharacter{2288}{\nsubseteq}
\DeclareUnicodeCharacter{2289}{\nsupseteq}
\DeclareUnicodeCharacter{228F}{\sqsubset}
\DeclareUnicodeCharacter{2290}{\sqsupset}
\DeclareUnicodeCharacter{2291}{\sqsubseteq}
\DeclareUnicodeCharacter{2292}{\sqsupseteq}
\DeclareUnicodeCharacter{2293}{\sqcap}
\DeclareUnicodeCharacter{2294}{\sqcup}
\DeclareUnicodeCharacter{2295}{\oplus}
\DeclareUnicodeCharacter{2296}{\ominus}
\DeclareUnicodeCharacter{2297}{\otimes}
\DeclareUnicodeCharacter{2298}{\oslash}
\DeclareUnicodeCharacter{2299}{\odot}
\DeclareUnicodeCharacter{229A}{\circledcirc}
\DeclareUnicodeCharacter{229B}{\circledast}
\DeclareUnicodeCharacter{229D}{\circleddash}
\DeclareUnicodeCharacter{22A2}{\vdash}
\DeclareUnicodeCharacter{22A3}{\dashv}
\DeclareUnicodeCharacter{22A4}{\top}
\DeclareUnicodeCharacter{22A5}{\bot}
\DeclareUnicodeCharacter{22AC}{\nvdash}
\DeclareUnicodeCharacter{22B2}{\vartriangleleft}
\DeclareUnicodeCharacter{22B3}{\vartriangleright}
\DeclareUnicodeCharacter{22C0}{\bigwedge}
\DeclareUnicodeCharacter{22C1}{\bigvee}
\DeclareUnicodeCharacter{22C2}{\bigcap}
\DeclareUnicodeCharacter{22C3}{\bigcup}
\DeclareUnicodeCharacter{22C5}{\cdot}
\DeclareUnicodeCharacter{22C6}{\star}
\DeclareUnicodeCharacter{22EE}{\vdots}
\DeclareUnicodeCharacter{22EF}{\cdots}
\DeclareUnicodeCharacter{22F1}{\ddots}
\DeclareUnicodeCharacter{27E6}{\llbracket}
\DeclareUnicodeCharacter{27E7}{\rrbracket}
\DeclareUnicodeCharacter{27E8}{\langle}
\DeclareUnicodeCharacter{27E9}{\rangle}
\DeclareUnicodeCharacter{27F8}{\Longleftarrow}
\DeclareUnicodeCharacter{27F9}{\Longrightarrow}

\DeclareUnicodeCharacter{2133}{\mathcal{M}}
\DeclareUnicodeCharacter{2191}{\uparrow}
\DeclareUnicodeCharacter{2193}{\downarrow}
\DeclareUnicodeCharacter{2197}{\nearrow}
\DeclareUnicodeCharacter{2198}{\searrow}
\DeclareUnicodeCharacter{21D1}{\Uparrow}

\usetikzlibrary{arrows, decorations.pathmorphing, matrix}
\tikzset{>=stealth', decoration={snake,amplitude=.3mm,segment length=2mm,post length=1mm}}

\DeclareMathOperator\tsbottom{bottom}
\DeclareMathOperator\dom{dom}
\DeclareMathOperator\down{down}
\DeclareMathOperator\equals{equals}

\DeclareMathOperator\id{id}
\DeclareMathOperator\push{push}

\DeclareMathOperator\set{set}
\DeclareMathOperator\sort{sort}
\DeclareMathOperator\stay{stay}
\DeclareMathOperator\type{type}
\DeclareMathOperator\up{up}

\newcommand{\TS}[1]{\mathrm{TS}(#1)}

\newcommand{\repeatcaption}[2]{%
  \renewcommand{\thefigure}{\ref{#1}}%
  \captionsetup{list=no}%
  \caption{#2, repeated from page \pageref{#1}.}%
}

\makeatletter
\def\thmt@refnamewithcomma #1#2#3,#4,#5\@nil{%
  \@xa\def\csname\thmt@envname #1utorefname\endcsname{#3}%
  \ifcsname #2refname\endcsname
    \csname #2refname\expandafter\endcsname\expandafter{\thmt@envname}{#3}{#4}%
  \fi
}
\makeatother

\theoremstyle{definition}

\declaretheorem[qed=$\Box$, name=Definition, refname={Definition,Definitions}, numberwithin=section]{definition}
\declaretheorem[qed=$\Box$, name=Construction, refname={Construction,Constructions}, sibling=definition]{construction}
\declaretheorem[qed=$\Box$, name=Example, refname={Example,Examples}, sibling=definition]{example}

\declaretheorem[name=Lemma, refname={Lemma,Lemmas}, sibling=definition]{lemma}

\declaretheorem[name=Theorem, refname={Theorem,Theorems}, sibling=definition]{theorem}
\declaretheorem[name=Proposition, refname={Proposition,Propositions}, sibling=definition]{proposition}

\allowdisplaybreaks

\title{\texorpdfstring{An automata characterisation for \\ multiple context-free languages\footnote{This is an extended version of a paper with the same title accepted at DLT 2016 which is available at \url{link.springer.com} and via the DOI \href{https://dx.doi.org/10.1007/978-3-662-53132-7_12}{\nolinkurl{10.1007/978-3-662-53132-7_12}}.}}{An automata characterisation for multiple context-free languages}}
\author{%
  Tobias Denkinger \\[.5em]
  \normalsize Faculty of Computer Science, \\
  \normalsize Technische Universität Dresden, \\%
  \normalsize 01062 Dresden, Germany \\%
  \normalsize \href{mailto:tobias.denkinger@tu-dresden.de}{\nolinkurl{tobias.denkinger@tu-dresden.de}}%
}

\begin{document}

\maketitle
\begin{abstract}
	We introduce tree stack automata as a new class of automata with storage and identify a restricted form of tree stack automata that recognises exactly the multiple context-free languages.
\end{abstract}


\tableofcontents
\clearpage

\section{Introduction}

Prominent classes of languages are often defined with the help of their generating mechanism, e.g.
	context-free languages are defined via context-free grammars,
	tree-adjoining languages via tree-adjoining grammars, and 
	indexed languages via indexed grammars. 
To achieve a better understanding of how languages from a specific language class can be recognised, it is natural to ask for an automaton model.
For context-free languages, this question is answered with pushdown automata~\cite{Cho62,Sch63}, yield languages of tree-adjoining grammars are recognised by embedded pushdown automata \cite[Sec.~3]{Vij88}, and indexed languages are recognised by nested stack automata~\cite{Aho69}.

Mildly context-sensitive grammars are currently prominent in natural language processing as they are able to express the non-projective constituents and dependencies that occur in natural languages~\cite{KuhSat09,Mai10}.
Multiple context-free grammars~\cite{SekMatFujKas91} describe many mildly context-sensitive grammars.
Yet, to the author's knowledge, there is no corresponding automaton model.
Thread automata \cite{Vil02,Vil02a}, introduced by Villemonte de la Clergerie to describe parsing strategies for mildly context-sensitive grammar formalisms, already come close to such an automaton model.
A construction of thread automata from ordered simple range concatenation grammars (which are equivalent to multiple context-free languages) was given \cite[Sec.~4]{Vil02}.
A construction for the converse direction as well as proofs of correctness, however, were not provided.

Based on the idea of thread automata, we introduce a new automaton model, tree stack automata, and formalise it using automata with storage~\cite{Sco67,Eng14} in the notation of Herrmann and Vogler~\cite{HerVog15}, see~\cref{sec:tree_stack_operator}.
Tree stack automata possess, in addition to the usual finite state control, the ability to manipulate a tree-shaped stack that has the tree's root at its bottom.
We find a restriction of tree stack automata that makes them equivalent to multiple context-free grammars and we give a constructive proof for this equivalence, see~\cref{sec:automata_equal_MCFG}.




\section{Preliminaries}
\label{sec:preliminaries}

In this section we fix some notation and briefly recall formalisms used throughout this paper.
%
We denote the set of natural numbers (including 0) by $ℕ$, $ℕ ∖ \{0\}$ by $ℕ_+$, and $\{1, …, n\}$ by $[n]$ for every $n ∈ ℕ$.
%
%
The reflexive, transitive closure of some endorelation $r$ is denoted as $r^*$.
%
For two sets $A$ and $B$, we denote the set of partial functions from $A$ to $B$ by $A → B$.
The operator $→$ shall be right associative.
Let $f: A → B$, $a ∈ A$, and $b ∈ B$.
The \emph{domain of $f$}, denoted by $\dom(f)$, is the subset of $A$ for which $f$ is defined.
If $\dom(f) = A$ we call $f$ \emph{total}.
We define $f[a ↦ b]$ as the partial function from $A$ to $B$ such that $f[a ↦ b](a) = b$ and $f[a ↦ b](a') = f(a')$ for every $a' ∈ \dom(f) ∖ \{a\}$.
We sometimes construe partial functions as relations in the usual manner.
%
Let $S$ be a countable set (of \emph{sorts}) and $s ∈ S$.
An \emph{$S$-sorted set} is a tuple $(B, \sort)$ where $B$ is a set and $\sort\colon B → S$ is total.
We denote the preimage of $s$ under $\sort$ by $B_s$ and abbreviate $(B, \sort)$ by $B$; $\sort$ will always be clear from the context.
%
Let $A$ be a set and $L ⊆ A^*$.
We call $L$ \emph{prefix-closed} if for every $w ∈ A^*$ and $a ∈ A$ we have that $wa ∈ L$ implies $w ∈ L$.
%
An \emph{alphabet} is a finite set (of \emph{symbols}).
Let $Γ$ be an alphabet.
The set of \emph{trees over $Γ$}, denoted by $\mathrm{T}_Γ$, is the set of partial functions from $ℕ_+^*$ to $Γ$ with finite and prefix-closed domain.
The usual definition of trees \cite[Sec.~2]{Gue83} additionally requires that for every $ρ ∈ ℕ_+^*$ and $n ≥ 2$: if $ρn$ is in the domain of a tree then $ρ(n-1)$ is as well; we drop this restriction here.

\subsection{Parallel multiple context-free grammars}

We fix a set $X = \{ x_i^j ∣ i, j ∈ ℕ_+ \}$ of \emph{variables}.
Let $Σ$ be an alphabet.
The set of \emph{composition representations over $Σ$} is the $(ℕ_+^* × ℕ_+)$-sorted set $\mathrm{RF}_Σ$ where for every $s_1, …, s_{\ell}, s ∈ ℕ_+$ we define
	$X_{(s_1 ⋯ s_{\ell}, s)} = \{x_i^j ∣ i ∈ [\ell], j ∈ [s_i] \} ⊆ X$ and
	\(
    (\mathrm{RF}_Σ)_{(s_1 ⋯ s_{\ell}, s)} = \{ [u_1, …, u_s]_{(s_1 ⋯ s_{\ell}, s)} ∣ u_1, …, u_s ∈ (Σ ∪ X_{(s_1 ⋯ s_{\ell}, s)})^* \}
  \) as a set of strings in which parentheses, brackets, commas, and the elements of $ℕ_+$, $Σ$, and $X_{(s_1⋯s_{\ell}, s)}$ are used as symbols.
Let $f = [u_1, …, u_s]_{(s_1 ⋯ s_{\ell}, s)} ∈ \mathrm{RF}_Σ$.
The \emph{composition function of $f$}, also denoted by $f$, is the function from $(Σ^*)^{s_1} × ⋯ × (Σ^*)^{s_{\ell}}$ to $(Σ^*)^s$ such that
\(
	f((w_1^1, …, w_1^{s_1}), …, (w_{\ell}^1, …, w_{\ell}^{s_{\ell}}))
		= (u_1', …, u_s')
\)
where $(u_1', …, u_s')$ is obtained from $(u_1, …, u_s)$ by replacing each occurrence of $x_i^j$ by $w_i^j$ for every $i ∈ [\ell]$ and $j ∈ [s_{\ell}]$.
The set of all composition functions for some composition representation over $Σ$ is denoted by $\mathrm{F}_Σ$.
From here on we no longer distinguish between composition representations and composition functions.
We define the \emph{fan-out of $f$} as $s$.
We call $f$ \emph{linear} (\emph{non-deleting}) if in $u_1 ⋯ u_s$ every element of $X$ occurs at most once (at least once, respectively).
The subscript is dropped from $f$ if its sort is clear from the context.

\begin{definition}
	A \emph{parallel multiple context-free grammar (short: PMCFG)} is a tuple
	\(
		G = (N, Σ, I, R)
	\) where
		$N$ is a finite $ℕ_+$-sorted set (of \emph{non-terminals}),
		$Σ$ is an alphabet (of \emph{terminals}),
		$I ⊆ N_1$ (\emph{initial non-terminals}), and
		$R ⊆ ⋃_{k, s, s_1, …, s_k ∈ ℕ} N_s × (\mathrm{F}_Σ)_{(s_1⋯s_k, s)} × (N_{s_1} × ⋯ × N_{s_k})$ is finite (\emph{rules}). 
\end{definition}

Let $G = (N, Σ, I, R)$ be a PMCFG.
A rule $(A, f, A_1⋯A_k) ∈ R$ is usually written as $A → f(A_1, …, A_k)$; it inherits its sort from $f$.
A PMCFG that only contains rules with linear composition functions is called a \emph{multiple context-free grammar (short: MCFG)}.
An MCFG that contains only rules of fan-out at most $k$ is called a $k$-MCFG.

For every $A ∈ N$, we recursively define the \emph{set of derivations in $G$ from $A$} as
\(
	\mathrm{D}_G(A) = \{ r(d_1, …, d_k) ∣ r = A → f(A_1, …, A_k) ∈ R, ∀i ∈ [k]: d_i ∈ \mathrm{D}_G(A_i) \}\text{.}
\)
The elements of $\mathrm{D}_G(A)$ can be construed as trees over $R$.
Let $d ∈ \mathrm{D}_G(A)$.
By projecting each rule in $d$ on its second component, we obtain a term over $\mathrm{F}_Σ$;
the \emph{tuple generated by $d$}, denoted by $⟦d⟧$, is obtained by evaluating this term.
We identify 1-tuples of strings with strings.
The \emph{set of (complete) derivations in $G$} is $\mathrm{D}_G = ⋃_{A ∈ N} \mathrm{D}_G(A)$ ($\mathrm{D}_G^{\text{c}} = ⋃_{S ∈ I} \mathrm{D}_G(S)$, respectively).
The \emph{language of $G$} is $L(G) = \{ ⟦d⟧ ∣ d ∈ \mathrm{D}_G^{\text{c}} \}$.

\subsection{Automata with storage}


\begin{definition}
	A \emph{storage type} is a tuple $S = (C, P, F, C_{\text{i}})$ where
		$C$ is a set (of \emph{storage configurations}),
		$P ⊆ \mathcal{P}(C)$ (\emph{predicates}),
		$F ⊆ C → C$ (\emph{instructions}), and
		$C_{\text{i}} ⊆ C$ (\emph{initial configurations}).
\end{definition}

\begin{definition}
	An \emph{automaton with storage} is a tuple $ℳ = (Q, S, Σ, q_{\text{i}}, c_{\text{i}}, δ, Q_{\text{f}})$ where
		$Q$ is a finite set (of \emph{states}),
		$S = (C, P, F, C_{\text{i}})$ is a storage type,
		$Σ$ is an alphabet (of \emph{terminals}),
		$q_{\text{i}} ∈ Q$ (\emph{initial state}),
		$c_{\text{i}} ∈ C_{\text{i}}$ (\emph{initial storage configuration}),
		$δ ⊆ Q × (Σ ∪ \{ε\}) × P × F × Q$ is finite (\emph{transitions}), and
		$Q_{\text{f}} ⊆ Q$ (\emph{final states}). 
\end{definition}

Let $ℳ = (Q, S, Σ, q_{\text{i}}, c_{\text{i}}, δ, Q_{\text{f}})$ be an automaton with storage and $S = (C, P, F, C_{\text{i}})$.
Let $τ = (q, ω, p, f, q') ∈ δ$ be a transition.
We call $q$ the \emph{source state of $τ$}, $p$ the \emph{predicate of $τ$}, $f$ the \emph{instruction of $τ$}, and $q'$ the \emph{target state of $τ$}.
A \emph{configuration of $ℳ$} is an tuple $(q, c, w)$ where $q ∈ Q$, $c ∈ C$, and $w ∈ Σ^*$.
We define the \emph{run relation with respect to $τ$} as the binary relation $⊢_τ$ on the set of configurations of $ℳ$ such that
\( (q, c, w) ⊢_τ (q', c', w') \) iff \( (w = ωw') ∧ (c ∈ p) ∧ (f(c) = c') \).
The \emph{set of runs in $ℳ$} is the smallest set $\mathrm{R}_ℳ ⊆ δ^*$ where for every $k ∈ ℕ$ and $τ_1, …, τ_k ∈ δ$, the string $θ = τ_1⋯τ_k$ is in $\mathrm{R}_ℳ$ if there are $q_0, …, q_k ∈ Q$, $c_0, …, c_k ∈ C$, and $ω_1, …, ω_k ∈ Σ ∪ \{ε\}$ such that
\( (q_0, c_0, ω_1⋯ω_k) ⊢_{τ_1} (q_1, c_1, ω_2⋯ω_k) ⊢_{τ_2} … ⊢_{τ_k} (q_k, c_k, ε) \);
we may then write $(q_0, c_0, ω_1⋯ω_k) ⊢_θ (q_k, c_k, ε)$ or $(q_0, c_0) ⊢_θ (q_k, c_k)$ or $⟦θ⟧ = ω_1⋯ω_k$.
The \emph{set of valid runs in $ℳ$}, denoted by $\mathrm{R}_ℳ^{\text{v}}$, contains exactly the runs $θ ∈ \mathrm{R}_ℳ$ where $(q_{\text{i}}, c_{\text{i}}) ⊢_θ (q, c)$ for some $q ∈ Q_{\text{f}}$ and $c ∈ C$.
For $θ ∈ \mathrm{R}_ℳ^{\text{v}}$ we say that $ℳ$ \emph{recognises} $⟦θ⟧$.
The \emph{language of $ℳ$} is $L(ℳ) = \{ ⟦θ⟧ ∣ θ ∈ \mathrm{R}_ℳ^{\text{v}} \}$.

%


\section{Tree stack automata}
\label{sec:tree_stack_operator}

Informally, a tree stack is a tree with a designated position in it.
The root of the tree serves as \emph{bottom-most symbol} and the leaves are \emph{top-most symbols}.
We allow the stack pointer to move \emph{downward} (i.e. to the parent) and \emph{upward} (i.e. to any child).
We may \emph{write} at any position except for the root.
We may also \emph{push} a symbol to any vacant child position of the current node.
Formally, for an alphabet $Γ$, a \emph{tree stack over $Γ$} is a tuple $(ξ[ε ↦ @], ρ)$ where $ξ ∈ \mathrm{T}_Γ$, $@ ∉ Γ$, and $ρ ∈ \dom(ξ) ∪ \{ε\}$.
The set of all tree stacks over $Γ$ is denoted by $\TS{Γ}$.
We define the following subsets (or predicates) of and partial functions on $\TS{Γ}$:
\begin{itemize}
	\item $\equals(γ) = \{ (ξ, ρ) ∈ \TS{Γ} ∣ ξ(ρ) = γ \}$ for every $γ ∈ Γ$ and
	\item $\tsbottom = \{ (ξ, ρ) ∈ \TS{Γ} ∣ ρ = ε \}$.
	\item $\id: \TS{Γ} → \TS{Γ}$ where $\id(ξ, ρ) = (ξ, ρ)$ for every $(ξ, ρ) ∈ \TS{Γ}$,
	\item $\push: ℕ_+ → Γ → \TS{Γ} → \TS{Γ}$ where $\push_n(γ)(ξ, ρ) = (ξ[ρn ↦ γ], ρn)$ for every $(ξ, ρ) ∈ \TS{Γ}$, $n ∈ ℕ_+$ with $ρn ∉ \dom(ξ)$, and $γ ∈ Γ$,
	\item $\up: ℕ_+ → \TS{Γ} → \TS{Γ}$ where $\up_n(ξ, ρ) = (ξ, ρn)$ for every $(ξ, ρ) ∈ \TS{Γ}$ and $n ∈ ℕ_+$ with $ρn ∈ \dom(ξ)$,
	\item $\down: \TS{Γ} → \TS{Γ}$ where $\down(ξ, ρn) = (ξ, ρ)$ for every $(ξ, ρn) ∈ \TS{Γ}$ with $n ∈ ℕ_+$, and
	\item $\set: Γ → \TS{Γ} → \TS{Γ}$ where $\set(γ)(ξ, ρ) = (ξ[ρ ↦ γ], ρ)$ for every $γ ∈ Γ$ and $(ξ, ρ) ∈ \TS{Γ}$ with $ρ ≠ ε$.
\end{itemize}
We may denote a tree stack $(ξ, ρ) ∈ \TS{Γ}$ by writing $ξ$ as a set and underlining the unique tuple of the form $(ρ, γ)$ in this set.
Consider for example a tree $ξ ∈ \mathrm{T}_{\{@, *, \#\}}$ with domain $\{ε, 2, 23\}$ such that $ξ: ε ↦ @,\; 2 ↦ *,\; 23 ↦ \#$.
We would then denote the tree stack $(ξ, 2) ∈ \TS{\{*, \#\}}$ by $\{(ε,@),\underline{(2,*)},(23,\#)\}$.

\begin{definition}
	Let $Γ$ be an alphabet.
	The \emph{tree stack storage with respect to $Γ$} is the storage type
	\((\TS{Γ}, P, F, \{\{\underline{(ε, @)}\}\})\), abbreviated by $\TS{Γ}$, where
	\begin{align*}
		P &= \{ \tsbottom, \equals(γ), \TS{Γ} ∣ γ ∈ Γ \} \text{ and} \\*
		F &= \{ \id, \push_n(γ), \up_n, \down, \set(γ) ∣ γ ∈ Γ, n ∈ ℕ \}\text{.} \tag*\qedhere
	\end{align*}
\end{definition}

We call automata with tree stack storage \emph{tree stack automata (short: TSA)}.
In a storage configuration $(ξ, ρ)$ of a TSA $ℳ$ we call $ξ$ the \emph{stack (of $ℳ$)} and $ρ$ the \emph{stack pointer (of $ℳ$)}.

\begin{example}\label{ex:automaton}
	Let $Σ = \{ \text{a}, \text{b}, \text{c}, \text{d} \}$ and $Γ = \{*, \#\}$.
	Consider the TSA
		\[ ℳ = \big( [5], \TS{Γ}, Σ, 1, \{\underline{(ε, @)}\}, δ, \{5\}\big)\]
	where $δ$ is shown in \cref{fig:automaton:run}.
	\Cref{fig:automaton:run} also shows the valid run $τ_1 τ_2 τ_3 τ_4 τ_5 τ_6 τ_7 τ_8 τ_9$ in $ℳ$ recognising $abcd$.
	The language of $ℳ$ is $L(ℳ) = \{ \text{a}^n \text{b}^n \text{c}^n \text{d}^n ∣ n ∈ ℕ \}$ and thus not context-free.
\end{example}

\newcommand\simpleAutomaton{%
  \begin{minipage}[t]{.45\textwidth}
    \setlength\abovedisplayskip{0pt}
    \begin{alignat*}{6}
      δ:\enspace &  & τ_1 & = \big(1 &  & , \text{a} &  & , \TS{Γ}      &  & , \push_1(*)  &  & , 1\big) \\*
                 &  & τ_2 & = \big(1 &  & , ε        &  & , \TS{Γ}      &  & , \push_1(\#) &  & , 2\big) \\
                 &  & τ_3 & = \big(2 &  & , ε        &  & , \equals(\#) &  & , \down       &  & , 2\big) \\
                 &  & τ_4 & = \big(2 &  & , \text{b} &  & , \equals(*)  &  & , \down       &  & , 2\big) \\
                 &  & τ_5 & = \big(2 &  & , ε        &  & , \tsbottom   &  & , \up_1       &  & , 3\big) \\
                 &  & τ_6 & = \big(3 &  & , \text{c} &  & , \equals(*)  &  & , \up_1       &  & , 3\big) \\
                 &  & τ_7 & = \big(3 &  & , ε        &  & , \equals(\#) &  & , \down       &  & , 4\big) \\
                 &  & τ_8 & = \big(4 &  & , \text{d} &  & , \equals(*)  &  & , \down       &  & , 4\big) \\*
                 &  & τ_9 & = \big(4 &  & , ε        &  & , \tsbottom   &  & , \id         &  & , 5\big)
    \end{alignat*}
  \end{minipage}
  \hfill
  \begin{minipage}[t]{.45\textwidth}
    \setlength\abovedisplayskip{0pt}
    \begin{alignat*}{5}
      &&&\big( 1&&, \{ \underline{(ε, @)} \}&&, \text{abcd} &&\big) \\
      &⊢_{τ_1}{}
      &&\big( 1&&, \{ (ε, @), \underline{(1, *)} \}&&, \text{bcd} &&\big) \\
      &⊢_{τ_2}{}
      &&\big( 2&&, \{ (ε, @), (1, *), \underline{(11, \#)} \}&&, \text{bcd} &&\big) \\
      &⊢_{τ_3}{}
      &&\big( 2&&, \{ (ε, @), \underline{(1, *)}, (11, \#) \}&&, \text{bcd} &&\big) \\
      &⊢_{τ_4}{}
      &&\big( 2&&, \{ \underline{(ε, @)}, (1, *), (11, \#) \}&&, \text{cd} &&\big) \\
      &⊢_{τ_5}{}
      &&\big( 3&&, \{ (ε, @), \underline{(1, *)}, (11, \#) \}&&, \text{cd} &&\big) \\
      &⊢_{τ_6}{}
      &&\big( 3&&, \{ (ε, @), (1, *), \underline{(11, \#)} \}&&, \text{d} &&\big) \\
      &⊢_{τ_7}{}
      &&\big( 4&&, \{ (ε, @), \underline{(1, *)}, (11, \#) \}&&, \text{d} &&\big) \\
      &⊢_{τ_8}{}
      &&\big( 4&&, \{ \underline{(ε, @)}, (1, *), (11, \#) \}&&, ε &&\big) \\
      &⊢_{τ_9}{}
      &&\big( 5&&, \{ \underline{(ε, @)}, (1, *), (11, \#) \}&&, ε &&\big)
    \end{alignat*}
  \end{minipage}%
}

\begin{figure}[t]
  \simpleAutomaton
  \caption{Set of transitions and a valid run in $ℳ$ (cf. \cref{ex:automaton}).}
  \label{fig:automaton:run}
\end{figure}

While $ℳ$ from the above example only uses a monadic stack, a TSA may also utilise branching as shown in the next example.

\begin{example}\label{ex:automaton:example2}
  Let again $Σ = \{ \text{a}, \text{b}, \text{c}, \text{d} \}$ and $Γ = \{*, \#\}$.
  Consider the TSA
  \[
    ℳ' = \big( [9], \TS{Γ}, Σ, 1,  \{\underline{(ε, @)}\}, δ', \{9\}\big)
  \]
  with $δ' = \{τ_1', …, τ_{15}'\}$ where
\begin{alignat*}{12}
	τ_{ 1}' & = \big(1 &  & , \text{a} &  & , \tsbottom   &  & , \push_1(*)  &  & , 2\big)\text{,} \qquad\qquad&
  τ_{ 9}' & = \big(5 &  & , ε        &  & , \tsbottom   &  & , \up_1       &  & , 6\big)\text{,} \\
  τ_{ 2}' & = \big(2 &  & , \text{a} &  & , \TS{Γ}      &  & , \push_1(*)  &  & , 2\big)\text{,} &
  τ_{10}' & = \big(6 &  & , \text{c} &  & , \equals(*)  &  & , \up_1       &  & , 6\big)\text{,} \\
  τ_{ 3}' & = \big(2 &  & , ε        &  & , \TS{Γ}      &  & , \push_1(\#) &  & , 3\big)\text{,} &
  τ_{11}' & = \big(6 &  & , ε        &  & , \equals(\#) &  & , \down       &  & , 7\big)\text{,} \\
  τ_{ 4}' & = \big(3 &  & , ε        &  & , \TS{Γ}      &  & , \down       &  & , 3\big)\text{,} &
  τ_{12}' & = \big(7 &  & , ε        &  & , \equals(*)  &  & , \down       &  & , 7\big)\text{,} \\
  τ_{ 5}' & = \big(3 &  & , \text{b} &  & , \tsbottom   &  & , \push_2(*)  &  & , 4\big)\text{,} &
  τ_{13}' & = \big(7 &  & , ε        &  & , \tsbottom   &  & , \up_2       &  & , 8\big)\text{,} \\
  τ_{ 6}' & = \big(4 &  & , \text{b} &  & , \TS{Γ}      &  & , \push_1(*)  &  & , 4\big)\text{,} &
  τ_{14}' & = \big(8 &  & , \text{d} &  & , \equals(*)  &  & , \up_1       &  & , 8\big)\text{, and} \\
  τ_{ 7}' & = \big(4 &  & , ε        &  & , \TS{Γ}      &  & , \push_1(\#) &  & , 5\big)\text{,} &
  τ_{15}' & = \big(8 &  & , ε        &  & , \equals(\#) &  & , \id         &  & , 9\big)\text{.} \\
  τ_{ 8}' & = \big(5 &  & , ε        &  & , \TS{Γ}      &  & , \down       &  & , 5\big)\text{,}
\end{alignat*}
Then $ℳ'$ recognises the languages $L' = \{\text{a}^i \text{b}^j \text{c}^i \text{d}^j ∣ i, j ∈ ℕ ∖ \{0\}\}$.
A valid run of $ℳ'$ on the word $\text{aabccd}$ is shown in \cref{fig:automaton:run2}.
\end{example}

\begin{figure}[t]
  \setlength\abovedisplayskip{0pt}
  \setlength\belowdisplayskip{0pt}
  \begin{alignat*}{5}
 & 
 &  & \big( 1 &  & , \{ \underline{(ε, @)} \}                                               &  & , \text{aabccd} &  & \big) \\
 & ⊢_{τ_1'}{}
 &  & \big( 2 &  & , \{ (ε, @), \underline{(1, *)} \}                                       &  & , \text{abccd}  &  & \big) \\
 & ⊢_{τ_2'}{}
 &  & \big( 2 &  & , \{ (ε, @), (1, *), \underline{(11, *)} \}                              &  & , \text{bccd}   &  & \big) \\
 & ⊢_{τ_3'}{}
 &  & \big( 3 &  & , \{ (ε, @), (1, *), (11, *), \underline{(111, \#)} \}                   &  & , \text{bccd}   &  & \big) \\
 & ⊢_{τ_4'τ_4'τ_4'}{}
 &  & \big( 3 &  & , \{ \underline{(ε, @)}, (1, *), (11, *), (111, \#) \}                   &  & , \text{bccd}   &  & \big) \\
 & ⊢_{τ_5'}{}
 &  & \big( 4 &  & , \{ (ε, @), (1, *), (11, *), (111, \#), \underline{(2, *)} \}           &  & , \text{ccd}    &  & \big) \\
 & ⊢_{τ_7'}{}
 &  & \big( 5 &  & , \{ (ε, @), (1, *), (11, *), (111, \#), (2, *), \underline{(21, \#)} \} &  & , \text{ccd}    &  & \big) \\
 & ⊢_{τ_8'τ_8'}{}
 &  & \big( 5 &  & , \{ \underline{(ε, @)}, (1, *), (11, *), (111, \#), (2, *), (21, \#) \} &  & , \text{ccd}    &  & \big) \\
 & ⊢_{τ_9'}{}
 &  & \big( 6 &  & , \{ (ε, @), \underline{(1, *)}, (11, *), (111, \#), (2, *), (21, \#) \} &  & , \text{ccd}    &  & \big) \\
 & ⊢_{τ_{10}'τ_{10}'}{}
 &  & \big( 6 &  & , \{ (ε, @), (1, *), (11, *), \underline{(111, \#)}, (2, *), (21, \#) \} &  & , \text{d}      &  & \big) \\
 & ⊢_{τ_{11}'}{}
 &  & \big( 7 &  & , \{ (ε, @), (1, *), \underline{(11, *)}, (111, \#), (2, *), (21, \#) \} &  & , \text{d}      &  & \big) \\
 & ⊢_{τ_{12}'τ_{12}'}{}
 &  & \big( 7 &  & , \{ \underline{(ε, @)}, (1, *), (11, *), (111, \#), (2, *), (21, \#) \} &  & , \text{d}      &  & \big) \\
 & ⊢_{τ_{13}'}{}
 &  & \big( 8 &  & , \{ (ε, @), (1, *), (11, *), (111, \#), \underline{(2, *)}, (21, \#) \} &  & , \text{d}      &  & \big) \\
 & ⊢_{τ_{14}'}{}
 &  & \big( 8 &  & , \{ (ε, @), (1, *), (11, *), (111, \#), (2, *), \underline{(21, \#)} \} &  & , ε             &  & \big) \\
 & ⊢_{τ_{15}'}{}
 &  & \big( 9 &  & , \{ (ε, @), (1, *), (11, *), (111, \#), (2, *), \underline{(21, \#)} \} &  & , ε             &  & \big)
  \end{alignat*}
  \caption{A valid run in $ℳ'$ (cf.~\cref{ex:automaton:example2}).}
  \label{fig:automaton:run2}
\end{figure}

%
\subsection{Restricted TSA}
Similar to Villemonte de la Clergerie \cite{Vil02}, we are interested in how often any specific position in the stack is reached from below.
For every TSA $ℳ$ we define $(c_ℳ(θ): ℕ_+^* → ℕ_+ ∣ θ ∈ R_ℳ^{\text{v}})$ as the family of total functions where
	$c_ℳ(ε)(ρ) = 0$ for every $ρ ∈ ℕ_+^*$, and
	for every $θτ ∈ R_ℳ^{\text{v}}$ with $τ ∈ δ$
		we have $c_ℳ(θτ) = c_ℳ(θ)$ if $τ$ has neither a $\push$- nor $\up$-instruction, and
		we have $c_ℳ(θτ) = c_ℳ(θ)[ρ ↦ c_ℳ(θ)(ρ) + 1]$ if $τ$ has a $\push$- or $\up$-instruction and $\{\underline{(ε, @)}\} ⊢_{θτ} (ξ, ρ)$ for some tree $ξ$.
We call $ℳ$ \emph{$k$-restricted} if $c_ℳ(θ)(ρ) ≤ k$ holds for every $θ ∈ R_ℳ^{\text{v}}$ and $ρ ∈ ℕ_+^*$.
Note that $ℳ$ from \cref{ex:automaton} and $ℳ'$ from \cref{ex:automaton:example2} are both 2-restricted.

Since (unrestricted) TSA can write at any position (except for $ε$) arbitrarily often, they can simulate Turing machines.
It is apparent that $1$-restricted TSA are exactly as powerful as pushdown automata.
The power of $k$-restricted TSA for $k ≥ 2$ is thus between the context-free and recursively enumerable languages.

%
%

\subsection{Normal forms}
We will see that loops that do not move the stack pointer as well as acceptance with non-$ε$ stack pointers can be removed.

Let $ℳ = (Q, \TS{Γ}, Σ, q_{\text{i}}, \{\underline{(ε, @)}\}, δ, Q_{\text{f}})$ be a TSA.
For each $q, q' ∈ Q$ and $γ, γ' ∈ Γ ∪ \{@\}$ we define $R_ℳ(q, q')|_{\stay}^{γ→γ'}$ as the set of runs $θ$ in $ℳ$ such that $θ$ only uses $\set$- or $\id$-instructions and there are tree stacks $(ξ, ρ), (ζ, ρ) ∈ \TS{Γ}$ with $ξ(ρ) = γ$, $ζ(ρ) = γ'$, and $(q, (ξ, ρ)) ⊢_θ (q', (ζ, ρ))$.

\begin{definition}
	We call a TSA $ℳ = (Q, \TS{Γ}, Σ, q_{\text{i}}, \{\underline{(ε, @)}\}, δ, Q_{\text{f}})$ \emph{cycle-free} if for every $q ∈ Q$ and $γ ∈ Γ ∪ \{@\}$ we have $R_ℳ(q,q)|_{\stay}^{γ→γ} = \{ε\}$.
\end{definition}

\begin{lemma}\label{lem:cycle-free}
	For every ($k$-restricted) TSA $ℳ$, there is a ($k$-restricted) cycle-free TSA $ℳ'$ such that $L(ℳ) = L(ℳ')$.
\end{lemma}
\begin{proof}[Proof idea.]\let\qed\relax
	Instead of performing all iterations of some loop $θ ∈ R_ℳ(q,q)|_{\stay}^{γ → γ} ∖ \{ε\}$ at the same position $ρ$ in the stack, we insert additional $\push$-instructions before each iteration of the loop.
	In order to find position $ρ$ again after the desired number of iterations, we write symbols $*$ or $\#$ before every $\push$, where a $*$ signifies that we have to perform at least two further $\down$-instructions to reach $ρ$ and $\#$ signifies that we will be at $ρ$ after one more $\down$-instruction.
	After returning to $ρ$, we enter a state $\tilde{q}$ that is equivalent to $q$ except that it prevents us from entering the loop again.
\end{proof}
\begin{proof}
	Let $ℳ = (Q, \TS{Γ}, Σ, q_{\text{i}}, \{\underline{(ε, @)}\}, δ, Q_{\text{f}})$ be a TSA and
  \[τ_1 ⋯ τ_n = (q_0, ω_1, p_1, f_1, q_1) ⋯ (q_{n-1}, ω_n, p_n, f_n, q_n)\]
  be a shortest element of $R_ℳ(q, q)|_{\stay}^{γ → γ} ∖ \{ε\}$ with \( q_0 = q = q_n \).
		
	Construct the automaton $ℳ' = (Q', \TS{Γ'}, Σ, q_{\text{i}}, \{\underline{(ε, @)}\}, δ', Q_{\text{f}}')$ where
    $Q' = Q ∪ \{q_0', …, q_{n-1}', q^{↑}, q^{↓}, \tilde{q}_0\}$,
    $q_0', …, q_{n-1}', q^{↑}, q^{↓}, \tilde{q}_0$ are pairwise different and not in $Q$,
    $Γ' = Γ ∪ \{*, \#\}$, $*$ and $\#$ are different and not in $Γ$,
    $Q_{\text{f}}' = Q_{\text{f}}$ if $q_0 ∉ Q_{\text{f}}$,
    $Q_{\text{f}}' = Q_{\text{f}} ∪ \{\tilde{q}_0\}$ if $q_0 ∈ Q_{\text{f}}$,
    $δ'$ contains the transition $(\bar{q}, ω, p', f, \hat{q})$ for every $(\bar{q}, ω, p, f, \hat{q}) ∈ δ ∖ \{τ_n\}$ where $p' = \TS{Γ'}$ if $p = \TS{Γ}$, and $p' = p$ otherwise,
    $δ'$ contains the transition $(\tilde{q}_0, ω, p', f, \hat{p})$ for every $(q_0, ω, p, f, \hat{q}) ∈ δ ∖ \{τ_1\}$ where $p' = \TS{Γ'}$ if $p = \TS{Γ}$, and $p' = p$ otherwise, and also
    $δ'$ contains transitions
    \begin{align*}
      \tilde{τ}_n &= (q_{n-1}, ω_n, p_n', f_n, \tilde{q}_0) \\*
      τ^{↑} &= (q, ε, \TS{Γ'}, \push_j(\#), q^{↑}) \,\text{,} \\
      τ_0' &= (q^{↑}, ε, \TS{Γ'}, \push_j(*), q_0') \,\text{,} \\
      τ_κ' &= (q_{κ-1}', ω_κ, \TS{Γ'}, \id, q_κ')
             \tag*{for every $κ ∈ [n-1]$,} \\
      τ_n' &= (q_{n-1}', ω_n, \TS{Γ'}, \id, q^{↑})\,\text{,} \\
      τ' &= (q^{↑}, ε, \TS{Γ'}, \id, q^{↓})\,\text{,} \\
      τ_{\text{a}}^{↓} &= (q^{↓}, ε, \equals(*), \down, q^{↓})\,\text{, and} \\*
      τ_{\text{b}}^{↓} &= (q^{↓}, ε, \equals(\#), \down, q)
    \end{align*}
    where $p_n' = \TS{Γ'}$ if $p_n = \TS{Γ}$ and $p_n' = p_n$ otherwise, and
  $j ∈ ℕ$ such that no $\push_j$-instruction occurs in \(δ\).
  By definition of the above transitions, we have
    \(⟦(τ_1 ⋯ τ_n)^{\ell}⟧ = ⟦τ^{↑} (τ_0' ⋯ τ_n')^{\ell} (τ_{\text{a}}^{↓})^{\ell} τ_{\text{b}}^{↓}⟧\)
  for every $\ell ∈ ℕ$ and hence for every valid run $θ$ in $ℳ$, there is a valid run $θ'$ in $ℳ'$ with $⟦θ⟧ = ⟦θ'⟧$. 
	We iterate the above construction until the automaton is cycle-free.
\end{proof}

\begin{definition}
	We say that a TSA $ℳ$ is in \emph{stack normal form} if the stack pointer of $ℳ$ is $ε$ whenever we reach a final state.
\end{definition}

\begin{lemma}\label{lem:stack-normal-form}
	For every ($k$-restricted) TSA $ℳ$, there is a ($k$-restricted) TSA $ℳ'$ in stack normal form such that $L(ℳ) = L(ℳ')$.
\end{lemma}
\begin{proof}[Proof idea.]\let\qed\relax
	We introduce a new state $q_{\text{f}}$ as the only final state and add transitions such that, beginning from any original final state, we may perform $\down$-instructions until the predicate $\tsbottom$ is satisfied and then enter state~$q_{\text{f}}$.
\end{proof}
\begin{proof}
	Let $ℳ = (Q, \TS{Γ}, Σ, q_{\text{i}}, \{\underline{(ε, @)}\}, δ, Q_{\text{f}})$ and $q_{\down}, q_{\text{f}} ∉ Q$.
	We construct an automaton $ℳ' = (Q ∪ \{q_{\down}, q_{\text{f}}\}, \TS{Γ}, Σ, q_{\text{i}}, \{\underline{(ε, @)}\}, δ', \{q_{\text{f}}\})$ where
	\begin{align*}
	δ' =
	δ ∪ \{ (q, ε, Γ, \id, q_{\down}) ∣ q ∈ Q_{\text{f}} \}
	  &∪ \{ (q_{\down}, ε, Γ, \down, q_{\down}) \} \\*
	  &∪ \{ (q_{\down}, ε, \tsbottom, \id, q_{\text{f}}) \}\;\text{.}
	\end{align*}
	Since $q_{\text{f}}$ is reachable from every element of $Q_{\text{f}}$ and every storage configuration without reading additional symbols, we have that $L(ℳ) = L(ℳ')$.
	Also $ℳ'$ is in stack normal form since $q_{\text{f}}$ can only be reached when the configuration satisfies the predicate $\tsbottom$.
	This construction preserves $k$-restrictedness since $δ' ∖ δ$ can not reach states from $Q$ and contains no additional $\push$ or $\up$-instructions.
\end{proof}

Note that $ℳ$ from \cref{ex:automaton} is cycle-free and in stack normal form whereas $ℳ'$ from \cref{ex:automaton:example2} is cycle-free but not in stack normal form.


\section{The equivalence of MCFG and restricted TSA}
\label{sec:automata_equal_MCFG}

\subsection{Every MCFG has an equivalent restricted TSA}
\label{sec:automata_supseteq_MCFL}

The following construction applies the idea of Villemonte de la Clergerie \cite[Sec.~4]{Vil02} to the case of parallel multiple context-free grammars where, additionally, we have to deal with copying, deletion, and permutation of argument components.
The overall idea is to incrementally guess for an input word $w$ a derivation $d$ of $G$ (that accepts $w$) on the stack while traversing the relevant components of the composition functions on the right-hand sides of already guessed rules (in $d$) left-to-right.
This specific traversal of the derivation tree is ensured using states and stack symbols that encode positions in the rules of $G$.%
\footnote{The control flow of our constructed automaton is similar to that of the treewalk evaluator for attribute grammars \cite[Sec.~3]{KenWar76}. The two major differences are that the treewalk evaluator also treats inherited attributes (which are not present in PMCFGs) and that our constructed automaton generates the tree on the fly (while the treewalk evaluator is already provided with the tree).}

\begin{construction}\label{con:PMCFL-toAutomaton}
Let $G = (N, Σ, I, R)$ be a PMCFG, $Γ = \{\Box\} ∪ R ∪ \bar{R}$, and
\( \bar{R} = \big\{ ⟨r, i, j⟩ ∣ r = A → [u_1, …, u_s](A_1, …, A_{\ell}) ∈ R, i ∈ [s], j ∈ \{0, …, \lvert u_i \rvert\} \big\}\,\text{.} \)
Intuitively, an element $⟨r, i, j⟩ ∈ \bar{R}$ stands for the position in $r$ right after the $j$-th symbol of the $i$-th component.
The \emph{automaton with respect to $G$} is $ℳ(G) = (Q, \TS{Γ}, Σ, \Box, \{\underline{(ε, @)}\}, \{ \Box \}, δ)$ where
	$Q = \{ q, q_+, q_- ∣ q ∈ \bar{R} ∪ \{\Box\} \}$ and $δ$ is the smallest set such that
for every $r = S → [u](A_1, …, A_{\ell}) ∈ R$ with $S ∈ I$, we have the transitions
	\begin{align*}
		\mathrm{init}(r)
			&= \big( \Box, ε, \TS{Γ}, \push_1(\Box), ⟨r, 1, 0⟩ \big)\,\text{,} \\
		\mathrm{suspend}_1(r, 1, \Box)
			&= \big( ⟨r, 1, \lvert u \rvert⟩, ε, \equals(\Box), \set(r), \Box_- \big)\,\text{, and} \\
		\mathrm{suspend}_2(\Box)
			&=\big( \Box_-, ε, \TS{Γ}, \down, \Box \big)\,\text{ in $δ$;}
\intertext{%
for every $r = A → [u_1, …, u_s](A_1, …, A_{\ell}) ∈ R$, $i ∈ [s]$, $j ∈ [\lvert u_i \rvert]$ where $σ ∈ Σ$ is the $j$-th symbol in $u_i$, we have the transition%
}
		\mathrm{read}(r, i, j)
			&= \big( ⟨r, i, j - 1⟩, σ, \TS{Γ}, \id, ⟨r, i, j⟩ \big)\,\text{ in $δ$,}
\intertext{%
and for every $r = A → [u_1, …, u_s](A_1, …, A_{\ell}) ∈ R$, $i ∈ [s]$, $j ∈ [\lvert u_i \rvert]$, $κ ∈ [\ell]$, $r' = A_κ → [v_1, …, v_{s'}](B_1, …, B_{\ell'}) ∈ R$, $m ∈ [s']$ where $x_κ^m ∈ X$ is the $j$-th symbol in $u_i$, we have the transitions (abbreviating $⟨r, i, j⟩$ by $q$)%
}
		\mathrm{call}(r, i, j, r')
			&= \big( ⟨r, i, j-1⟩, ε, \TS{Γ}, \push_κ(q), ⟨r', m, 0⟩ \big)\,\text{,} \\
		\mathrm{resume}_1(r, i, j)
			&= \big( ⟨r, i, j-1⟩, ε, \TS{Γ}, \up_κ, q_+ \big)\,\text{,} \\
		\mathrm{resume}_2(r, i, j, r')
			&= \big( q_+, ε, \equals(r'), \set(q), ⟨r', m, 0⟩ \big)\,\text{,} \\
		\mathrm{suspend}_1(r', m, q)
			&= \big( ⟨r', m, \lvert v_m \rvert⟩, ε, \equals(q), \set(r'), q_- \big)\,\text{, and} \\
		\mathrm{suspend}_2(q)
			&=\big( q_-, ε, \TS{Γ}, \down, q \big)\,\text{ in $δ$}\tag*\qedhere
	\end{align*}
\end{construction}

Let us abbreviate a run $\mathrm{suspend}_1(r', m, q)\,\mathrm{suspend}_2(q)$ by $\mathrm{suspend}(r', m, q)$ and a run $\mathrm{resume}_1(r, i, j)\,\mathrm{resume}_2(r, i, j, r')$ by $\mathrm{resume}(r, i, j, r')$.


\begin{example}\label{ex:PMCFL_toAutomaton_subseteq}
	Consider the MCFG $G = (\{S, A, B\}, \{a, b, c, d\}, \{S\}, R)$ where
	\begin{align*}
		R\colon\enspace
		r_1 &= S → [x_1^1 x_2^1 x_1^2 x_2^2](A,B) &
		r_2 &= A → [\text{a} x_1^1, \text{c} x_1^2](A) &
		r_3 &= A → [ε, ε]() \\*&&
		r_4 &= B → [\text{b} x_1^1, \text{d} x_1^2](B) &
		r_5 &= B → [ε, ε]()\,\text{.}
	\end{align*}
	Then $L(G) = \{ \text{a}^i \text{b}^j \text{c}^i \text{d}^j ∣ i, j ∈ ℕ \}$.
  \Cref{fig:PMCFL_toAutomaton_subseteq} shows that $ℳ(G)$ recognises $\text{bd}$.
\end{example}

\begin{quote}
  \emph{For the rest of \cref{sec:automata_supseteq_MCFL}, let $G = (N, Σ, I, R)$ and $\bar{R}$ be defined as in \cref{con:PMCFL-toAutomaton}.}
\end{quote}

\begin{figure}[tb]
	\begin{alignat*}{4}
		&&&\big( \Box&&, \{\underline{(ε, @)}\} &&\big) \\*
		&⊢_{\mathrm{init}(r_1)}{}
			&&\big( ⟨r_1, 1, 0⟩&&, \{ (ε,@), \underline{(1,\Box)} \} &&\big) \\
		&⊢_{\mathrm{call}(r_1, 1, 1, r_3)}{}
			&&\big( ⟨r_3, 1, 0⟩&&, \{ (ε,@), (1,\Box), \underline{(11,⟨r_1, 1, 1⟩)} \} &&\big) \\
		&⊢_{\mathrm{suspend}(r_3, 1, ⟨r_1, 1, 1⟩)}{}
			&&\big( ⟨r_1, 1, 1⟩&&, \{ (ε,@), \underline{(1,\Box)}, (11,r_3) \} &&\big) \\
		&⊢_{\substack{\mathrm{call}(r_1, 1, 2, r_4) \\ \mathrm{read}(r_4, 1, 1)}}{}
			&&\big( ⟨r_4, 1, 1⟩&&, \{ (ε,@), (1,\Box), (11,r_3), \underline{(12,⟨r_1, 1, 2⟩)} \} &&\big) \\
		&⊢_{\mathrm{call}(r_4, 1, 2, r_5)}{}
			&&\big( ⟨r_5, 1, 0⟩&&, \{ (ε,@), (1,\Box), (11,r_3), (12,⟨r_1, 1, 2⟩), \underline{(121, ⟨r_4, 1, 2⟩)} \} &&\big) \\
		&⊢_{\mathrm{suspend}(r_5, 1, ⟨r_4, 1, 2⟩)}{}
			&&\big( ⟨r_4, 1, 2⟩&&, \{ (ε,@), (1,\Box), (11,r_3), \underline{(12,⟨r_1, 1, 2⟩)}, (121, r_5) \} &&\big) \\
		&⊢_{\mathrm{suspend}(r_4, 1, ⟨r_1, 1, 2⟩)}{}
			&&\big( ⟨r_1, 1, 2⟩&&, \{ (ε,@), \underline{(1,\Box)}, (11,r_3), (12, r_4), (121, r_5) \} &&\big) \\
		&⊢_{\mathrm{resume}(r_1, 1, 3, r_3)}{}
			&&\big( ⟨r_3, 2, 0⟩&&, \{ (ε,@), (1,\Box), \underline{(11,⟨r_1, 1, 3⟩)}, (12, r_4), (121, r_5) \} &&\big) \\
		&⊢_{\mathrm{suspend}(r_3, 2, ⟨r_1, 1, 3⟩)}{}
			&&\big( ⟨r_1, 1, 3⟩&&, \{ (ε,@), \underline{(1,\Box)}, (11,r_3), (12, r_4), (121, r_5) \} &&\big) \\
		&⊢_{\substack{\mathrm{resume}(r_1, 1, 4, r_4) \\ \mathrm{read}(r_4, 2, 1)}}{}
			&&\big( ⟨r_4, 2, 1⟩&&, \{ (ε,@), (1,\Box), (11,r_3), \underline{(12, ⟨r_1, 1, 4⟩)}, (121, r_5) \} &&\big) \\
		&⊢_{\mathrm{resume}(r_4, 2, 2, r_5)}{}
			&&\big( ⟨r_5, 2, 0⟩&&, \{ (ε,@), (1,\Box), (11,r_3), (12, ⟨r_1, 1, 4⟩), \underline{(121, ⟨r_4, 2, 2⟩)} \} &&\big) \\
		&⊢_{\mathrm{suspend}(r_5, 2, ⟨r_4, 2, 2⟩)}{}
			&&\big( ⟨r_4, 2, 2⟩&&, \{ (ε,@), (1,\Box), (11,r_3), \underline{(12, ⟨r_1, 1, 4⟩)}, (121, r_5) \} &&\big) \\
		&⊢_{\mathrm{suspend}(r_4, 2, ⟨r_1, 1, 4⟩)}{}
			&&\big( ⟨r_1, 1, 4⟩&&, \{ (ε,@), \underline{(1,\Box)}, (11,r_3), (12, r_4), (121, r_5) \} &&\big) \\
		&⊢_{\mathrm{suspend}(r_1, 1, \Box)}{}
			&&\big( \Box&&, \{ \underline{(ε,@)}, (1, r_1), (11,r_3), (12, r_4), (121, r_5) \} &&\big)
	\end{alignat*}
	\caption{Run of $ℳ(G)$ that recognises $\text{bd}$ (cf. \cref{ex:PMCFL_toAutomaton_subseteq}). The symbols $\text{b}$ and $\text{d}$ are read by $\mathrm{read}(r_4, 1, 1)$ and $\mathrm{read}(r_4, 2, 1)$, respectively, all other transitions in this run read $ε$.}
	\label{fig:PMCFL_toAutomaton_subseteq}
\end{figure}

\begin{lemma}\label{lem:MCFL_Automaton_restricted}
	The TSA $ℳ(G)$ is $k$-restricted if $G$ is a $k$-MCFG.
\end{lemma}
\begin{proof}
	Let $G = (N, Σ, I, R)$.
	Consider some arbitrary position $ρ ∈ ℕ^*$ and number $κ ∈ ℕ$.
	Position $ρκ$ can only be reached from below if the current stack pointer is at position $ρ$ and if we either execute
	the transition \( \mathrm{call}(r, i, j, r') \) or
	the transition \( \mathrm{resume}_1(r, i, j) \)
	for some
		$r = A → [u_1, …, u_s](A_1, …, A_{\ell}) ∈ R$,
		$r' = A_κ → [v_1, …, v_{s'}](B_1, …, B_{\ell'}) ∈ R$,
		$i ∈ [s]$,
		$j ∈ [\lvert u_i \rvert]$, and
		$m ∈ [s']$ with $(u_i)_j = x_κ^m$.
	For those transitions to be applicable, the automaton has to be in state $⟨r, i, j⟩$.
	Therefore, there are exactly as many states from which we can reach position $ρi$ as there are occurrences of elements of $\{x_κ^1, …, x_κ^{s'}\}$ in the string $u_1⋯u_s$.
	Since $[u_1, …, u_s]$ is linear, the number of such occurrences is smaller or equal to $s'$ and (since $G$ is a $k$-MCFG) also smaller or equal to $k$.
	It is easy to see that in the part of the run where the stack pointer is never below $ρ$, the states $⟨r, i, 1⟩, …, ⟨r, i, \lvert u_i \rvert⟩$ occur in that order whenever the stack pointer is at $ρ$ and, in particular, none of those states occur twice.
	Therefore, we have that $c_{ℳ(G)}(θ)(ρκ) ≤ k$ for every run $θ$ and, since $ρ$ and $κ$ were chosen arbitrarily, we have that for any non-empty position $ρ' ≠ ε$ and every run $θ$ holds $c_{ℳ(G)}(θ)(ρ') ≤ k$.
	Since the position $ε$ can never be entered from below, we have that $ℳ(G)$ is $k$-restricted.
\end{proof}

\begin{lemma}\label{lem:PMCFL_toAutomaton_subseteq}
	$L(G) ⊆ L(ℳ(G))$.
\end{lemma}
\begin{proof}
	For every $A ∈ N$ and every derivation $d = r(d_1, …, d_m) ∈ \mathrm{D}_G(A)$ where $\sort(r) = (s_1 ⋯ s_m, s)$ and $r = A → [u_1, …, u_s](B_1, …, B_m)$, we recursively construct a tuple $(θ^1, …, θ^s)$ of runs in $ℳ(G)$.
	For the derivations $d_1, …, d_m$ we already have the tuples $(θ_1^1, …, θ_1^{s_1}), …, (θ_m^1, …, θ_m^{s_m})$, respectively.
	For every $κ ∈ [s]$, let $u_κ = ω_1⋯ω_{\ell}$ where $ω_1, …, ω_{\ell} ∈ Σ ∪ X$.
	We define $θ_κ = ω_1' ⋯ ω_{\ell}'$ as the run in $ℳ(G)$ such that for every $κ' ∈ [\ell]$, we have that
		$ω_{κ'}' = \mathrm{read}(r, κ, κ')$ if $ω_{κ'} ∈ Σ$,
		$ω_{κ'}' = \mathrm{call}(r, κ, κ', r')\, θ_i^1\, \mathrm{suspend}(r', 1, ⟨r, κ, κ'⟩)$ if $ω_{κ'} = x_i^1$ for some $i ≥ 1$, and
		$ω_{κ'}' = \mathrm{resume}(r, κ, κ', r')\, θ_i^j\, \mathrm{suspend}(r', j, ⟨r, κ, κ'⟩)$ if $ω_{κ'} = x_i^j$ for some $i ≥ 1$ and $j ≥ 2$,
		where $r' = d_i(ε)$.
	We can prove by structural induction on $d$ that $⟦d⟧ = (⟦θ^1⟧, …, ⟦θ^s⟧)$.
	If $d ∈ \mathrm{D}_G^{\text{c}}$, then $s$ is 1 and hence the valid run $\mathrm{init}(r)\, θ^1\, \mathrm{suspend}(r, 1, \Box)$ recognises exactly $⟦d⟧$.
\end{proof}

\begin{lemma}\label{lem:phi(rho)}
	Let $τ_1, …, τ_n ∈ δ$ with $θ = τ_1⋯τ_n ∈ R_{ℳ(G)}$ and let $ρ ∈ ℕ_+^* ∖ \{ε\}$.
	There is a rule $φ_θ(ρ)$ in $G$ such that, during the run $θ$, the automaton $ℳ(G)$ is in some state $⟨φ_θ(ρ), i, j⟩ ∈ \bar{R}$ whenever the stack pointer is at $ρ$.
\end{lemma}
\begin{proof}
	The rule $φ_θ(ρ)$ is selected when $ρ$ is first reached (with $\mathrm{call}$).
	Then whenever we enter $ρ$ with $\mathrm{resume}$, a previous $\mathrm{suspend}_1$ has stored $φ_θ(ρ)$ at position $ρ$ and $\mathrm{resume}_2$ enforces the claimed property.
	The claimed property is preserved by $\mathrm{read}$.
	And whenever we enter $ρ$ with $\mathrm{suspend}$, a previous $\mathrm{call}$ or $\mathrm{resume}_2$ has stored an appropriate state in the stack and $\mathrm{suspend}$ merely jumps back to that state, observing the claimed property.
\end{proof}

Examining the form of runs in $ℳ(G)$ (\cref{con:PMCFL-toAutomaton}) and using \cref{lem:phi(rho)} we observe:

\begin{lemma}\label{lem:states}
	Let $τ, τ' ∈ δ$, $q, q', q'' ∈ Q$, $ξ, ξ', ξ'' ∈ \TS{Γ}$, $ρ ∈ ℕ_+^*$, $i ∈ ℕ_+$, and $φ_θ(ρi)$ be of the form $A → [u_1, …, u_s](A_1, …, A_{\ell})$.
  Then:
	\begin{enumerate}
		\item\label{item:states_a}
			If $(q', (ξ', ρ)) ⊢_τ (q, (ξ, ρi))$ with $q ∈ \bar{R}$, then $q = ⟨φ_θ(ρi), j, 0⟩$ for some $j ∈ [s]$ and $τ$ must be either an $\mathrm{init}$- or $\mathrm{call}$-transition.
		\item\label{item:states_a1}
			If $(q'', (ξ'', ρ)) ⊢_τ (q', (ξ', ρi)) ⊢_{τ'} (q, (ξ, ρi))$ with $q' ∈ \{ q_+ ∣ q ∈ \bar{R} \}$, then $q = ⟨φ_θ(ρi), j, 0⟩$ for some $j ∈ [s]$, $τ$ is a $\mathrm{resume}_1$-transition, and $τ'$ is a $\mathrm{resume}_2$-transition.
		\item\label{item:states_b}
			If $(q, (ξ, ρi)) ⊢_τ (q', (ξ', ρi)) ⊢_{τ'} (q'', (ξ'', ρ))$, then $q = ⟨φ_θ(ρi), j, \lvert u_j \rvert⟩$ for some $j ∈ [s]$, $τ$ is a $\mathrm{suspend}_1$-transition, and $τ'$ is a $\mathrm{suspend}_2$-transition.
	\end{enumerate}
\end{lemma}
\begin{proof}
	(for~\ref{item:states_a},~\ref{item:states_a1}, and~\ref{item:states_b})~%
	The first projection of $q$ is $φ_θ(ρi)$ due to \cref{lem:phi(rho)}.
	
	(for~\ref{item:states_a})~%
	We only move the stack pointer to a child position and simultaneously go to a state from the set $\bar{R}$ when making a $\mathrm{init}$ or a $\mathrm{call}$ transition.
	From the definition of $\mathrm{init}$ and $\mathrm{call}$ transitions we know that the third projection of $q$ is 0.
	
	(for~\ref{item:states_a1})~%
	We only move the stack pointer to a child position and simultaneously go to a state from the set $\{ q_+ ∣ q ∈ \bar{R} \}$ when making a $\mathrm{resume}_1$ transition.
	Every $\mathrm{resume}_1$ transition is followed by a $\mathrm{resume}_2$ transition.
	From the definition of $\mathrm{resume}_2$ transitions we know that the third projection of $q$ is 0.
	
	(for~\ref{item:states_b})~%
	We only move the stack pointer to a parent position when making a $\mathrm{suspend}_2$ transition.
	Every $\mathrm{suspend}_2$ transition is preceded by a $\mathrm{suspend}_1$ transition.
	From the definition of $\mathrm{suspend}_1$ transitions we know that the third projection of $q$ is $\lvert u_j \rvert$.
\end{proof}

\begin{lemma}\label{lem:PMCFL_toAutomaton_supseteq}
	$L(G) ⊇ L(ℳ(G))$ if $G$ only has productive non-terminals.
\end{lemma}
\begin{proof}
	For every run $θ ∈ R_{ℳ(G)}$ we define $φ_θ': ℕ^* → R$ by $φ_θ'(ρ) = φ_θ(1ρ)$ for every $ρ ∈ ℕ_+^*$ with $1ρ ∈ \dom(φ_θ)$ (cf. \cref{lem:phi(rho)}).
	Then $φ_θ'$ is a tree.
	One could show for every $d ∈ \mathrm{D}_G$ with $d ⊇ φ_θ'$ by structural induction on $φ_θ'$ that for every $ρ ∈ \dom(φ_θ')$ and every maximal interval $[a, b]$ where $ρ_a, …, ρ_b$ have prefix $ρ$, we have
	\( ⟦τ_a⋯τ_b⟧ = ⟦d|_ρ⟧_m \)
	with $q_a = ⟨φ_θ'(ρ), m, 0⟩$ for some $m ∈ ℕ_+$.
	Let us call this property ($\dagger$).
	Let $τ_1, …, τ_n ∈ δ$ with $θ = τ_1⋯τ_n ∈ R_{ℳ(G)}^{\text{v}}$.
	Consider the run
	\( (\Box, (@, ε)) ⊢_{τ_1} (q_1, (ξ_1, 1ρ_1)) ⊢_{τ_2} … ⊢_{τ_{n-1}} (q_{n-1}, (ξ_{n-1}, 1ρ_{n-1})) ⊢_{τ_n} (\Box, (ξ_n, ε)) \).
  By ($\dagger$) we obtain that $⟦τ_2⋯τ_{n-1}⟧ = ⟦d⟧$.
	By \cref{lem:states} and the fact that only an $\mathrm{init}$-transition may start from $\Box$ we obtain that $τ_1$ is an $\mathrm{init}$-transition and $τ_n$ is a $\mathrm{suspend}_2$-transition.
	Thus $⟦τ_1⟧ = ε = ⟦τ_n⟧$ and therefore $⟦θ⟧ = ⟦d⟧$.

  It remains to proof ($\dagger$).
	For this we will denote the $i$-th component of the tuple generated by a derivation $d$ as $⟦d⟧_i$.
	Let $φ_θ'(ρ) = A → [u_1, …, u_s](A_1, …, A_{\ell})$.
	Consider the run
	\(
		(q_{a-1}, (ξ_{a-1}, 1ρ_{a-1}))
		⊢_{τ_a} (q_a, (ξ_a, 1ρ_a))
		⊢_{τ_{a+1}} …
		⊢_{τ_b} (q_b, (ξ_b, 1ρ_b))
	\).
	Since $[a, b]$ is maximal and a transition can add at most one symbol to the stack pointer, we know that $ρ_a = ρ = ρ_b$.
	By \cref{lem:states} we also know that $q_a = ⟨φ_θ'(ρ), m, 0⟩$ and $q_b = ⟨φ_θ'(ρ), m, \lvert u_m \rvert⟩$ for some $m ∈ [s]$.
	We now define the strings $w_1, …, w_{\lvert u_m \rvert}$ for every $i ∈ [\lvert u_m \rvert]$
	\begin{enumerate}
		\item\label{case:PMCFL_toAutomaton_supseteq:Sigma}
			as $w_i = σ$ if $(u_m)_i = σ$ for some $σ ∈ Σ$ and
		\item\label{case:PMCFL_toAutomaton_supseteq:X}
			as $w_i = ⟦d|_{ρκ}⟧_j$ if $(u_m)_i = x^j_κ$ for some $x^j_κ ∈ X$.
	\end{enumerate}
	From \cref{con:PMCFL-toAutomaton} we know that if we are in Case~\ref{case:PMCFL_toAutomaton_supseteq:Sigma}, then the only possibility to continue from state $⟨φ_θ'(ρ), m, i - 1⟩$ is to use the transition $\mathrm{read}(φ_θ'(ρ), m, i)$.
	We then end up in state $⟨φ_θ'(ρ), m, i⟩$.
	If we are in Case~\ref{case:PMCFL_toAutomaton_supseteq:X} and in state $⟨φ_θ'(ρ), m, i - 1⟩$, then, due to \cref{con:PMCFL-toAutomaton,lem:phi(rho)}, we can only continue with either (depending on the current stack)
	\(
		\mathrm{call}(φ_θ'(ρ), m, i, φ_θ'(ρκ))
	\) or
	\(
		\mathrm{resume}(φ_θ'(ρ), m, i, φ_θ'(ρκ))
	\).
	By induction hypothesis we know that after executing either of the above runs, the automaton will recognise $⟦d|_{ρκ}⟧_j = w_i$ and set the stack pointer to $ρ$.
	Then (by \cref{con:PMCFL-toAutomaton}) the automaton is in state $⟨φ_θ'(ρ), m, i + 1⟩$.
	Repeating the step above eventually brings $ℳ(G)$ to the state $⟨φ_θ'(ρ), m, \lvert u_m \rvert⟩$ where only some $\mathrm{suspend}$-transition is applicable.
	Thus $⟦τ_a⋯τ_b⟧ = w_1⋯w_{\lvert u_m \rvert} = ⟦d|_ρ⟧_m$.
\end{proof}

\begin{proposition}\label{prop:PMCFL_toAutomaton}
	$L(G) = L(ℳ(G))$ if $G$ only has productive non-terminals.
\end{proposition}
\begin{proof}
	The claim follows directly from \cref{lem:PMCFL_toAutomaton_subseteq,lem:PMCFL_toAutomaton_supseteq}.
\end{proof}

\paragraph{$ℳ(G)$ is almost a parser for $G$.}
Let $(ξ, ε)$ be a storage configuration of $ℳ(G)$ after recognising some word $w$ and let $ξ|_1$ be the first subtree of $ξ$, defined by the equation $ξ|_1(ρ) = ξ(1ρ)$.
Then every complete derivation $d$ in $G$ with $ξ|_1 ⊆ d$ generates $w$.
If $G$ only contains rules with non-deleting composition functions, we even have that $ξ|_1$ \emph{is} a derivation in $G$ generating $w$.
In \cref{fig:PMCFL_toAutomaton_subseteq}, for example, we see that $r_1(r_3, r_4(r_5))$ is a derivation of $bd$ in $G$ (cf.~\cref{ex:PMCFL_toAutomaton_subseteq}).

\subsection{Every restricted TSA has an equivalent MCFG}
\label{sec:automata_subseteq_MCFL}

We construct an MCFG $G'(ℳ)$ that recognises the valid runs of a given automaton $ℳ$, and then use the closure of MCFGs under homomorphisms.
A tuple of runs $(θ_1, …, θ_m)$ can be derived from non-terminal $⟨q_1, q_1', …, q_m, q_m'; γ_0, …, γ_m⟩$ iff the runs $θ_1, …, θ_m$ all return to the stack position they started from and never go below it, and $θ_i$ starts from state $q_i$ and stack symbol $γ_{i-1}$ and ends with $q_i'$ and $γ_i$ for every $i ∈ [m]$.
We start with an example.

\begin{figure}[t]
\simpleAutomaton
\repeatcaption{fig:automaton:run}{Set of transitions and a valid run in $ℳ$, cf.~\cref{ex:automaton}}
\end{figure}

\begin{example}
	Recall the TSA $ℳ$ from \cref{ex:automaton} (also cf. \cref{fig:automaton:run}).
  Note that $ℳ$ is cycle-free and in stack normal form.
	Let us consider position $ε$ of the stack.
	The only transitions applicable there are $τ_1$, $τ_2$, $τ_5$, and $τ_9$.
	Clearly, every valid run in $ℳ$ starts with $τ_1$ or $τ_2$ and ends with $τ_9$, every $τ_5$ must be preceded by $τ_4$ or $τ_3$, and every $τ_9$ must be preceded by $τ_8$ or $τ_7$.
	Thus each valid run in $ℳ$ is either of the form
	\(
		θ = τ_1 θ_1 τ_4 τ_5 θ_2 τ_8 τ_9
	\) or
	\(
		θ' = τ_2 θ_1' τ_3 τ_5 θ_2' τ_7 τ_9
	\)
	for some runs $θ_1$, $θ_2$, $θ_1'$, and~$θ_2'$.
	The target state of $τ_1$ is 1 and the source state of $τ_4$ is 2.
	Also $τ_1$ pushes a $*$ to position 1 and the predicate of $τ_4$ accepts only~$*$.
	Thus $θ_1$ must go from state 1 to 2 and from stack symbol $*$ to $*$ at position 1.
	Similarly, we obtain that $θ_2$, $θ_1'$, and $θ_2'$ go from state 3 to 4, 2 to 2, and 3 to 3, respectively, and from stack symbol $*$ to $*$, $\#$ to $\#$, and $\#$ to $\#$, respectively, at position 1.
	The runs $θ_1$ and $θ_2$ are linked since they are both executed while the stack pointer is in the first subtree of the stack; the same holds for $θ_1'$ and $θ_2'$.

	Clearly, linked runs need to be produced by the same non-terminal.
	For the pair $(θ_1, θ_2)$ of linked runs, we have the non-terminal $⟨1,2,3,4;*,*,*⟩$ 
	and for $(θ_1', θ_2')$ we have $⟨2,2,3,3;\#,\#,\#⟩$.
	Since $θ$ and $θ'$ go from state 1 to 5 and from storage symbol $@$ to $@$, we have the rules
	\begin{align*}
		⟨1,5;@,@⟩ &→ \big[ τ_1 x_1^1 τ_4 τ_5 x_1^2 τ_8 τ_9 \big]\big( ⟨1,2,3,4;*,*,*⟩ \big)\text{ and} \\*
		⟨1,5;@,@⟩ &→ \big[ τ_2 x_1^1 τ_3 τ_5 x_1^2 τ_7 τ_9 \big]\big( ⟨2,2,3,3;\#,\#,\#⟩ \big)\text{ in $G'(ℳ)$.}
	\end{align*}
	
	Next, we explore the non-terminal $⟨1,2,3,4;*,*,*⟩$, i.e. we need a run that goes from state 1 to 2 and from storage symbol $*$ to $*$ and another run that goes from state 3 to 4 and from storage symbol $*$ to $*$.
	There are only two kinds of suitable pairs of runs: $\big( τ_1 θ_1 τ_4, τ_6 θ_2 τ_8 \big)$ and $\big( τ_2 θ_1' τ_3, τ_6 θ_2' τ_7 \big)$ for some runs $θ_1$, $θ_2$, $θ_1'$, and~$θ_2'$.
	The runs $θ_1$, $θ_2$, $θ_1'$, and $θ_2'$ of this paragraph then have the same state and storage behaviour as in the previous paragraph and we have rules
	\begin{align*}
		⟨1,2,3,4;*,*,*⟩ &→ \big[ τ_1 x_1^1 τ_4, τ_6 x_1^2 τ_8 \big]\big( ⟨1,2,3,4;*,*,*⟩ \big) \text{ and} \\*
		⟨1,2,3,4;*,*,*⟩ &→ \big[ τ_2 x_1^1 τ_3, τ_6 x_1^2 τ_7 \big]\big( ⟨2,2,3,3;\#,\#,\#⟩ \big) \text{ in $G'(ℳ)$.}
	\end{align*}
	
	For non-terminal $⟨2,2,3,3;\#,\#,\#⟩$, we may only take the pair of empty runs and thus have the rule
	\(
		⟨2,2,3,3;\#,\#,\#⟩ → \big[ ε, ε \big]\big(\big)
	\) in $G'(ℳ)$.
\end{example}

For all $q, q' ∈ Q$, $γ, γ' ∈ Γ$, and $j ∈ ℕ_+$ we define the following sets:
\begin{align*}
	δ(q, q')|_{\up_j}^{γ↗∙}
		&= \{ (q, ω, p, \up_j, q') ∈ δ ∣ γ ∈ p \}\,\text{,} \\
	δ(q, q')|_{\push_j}^{γ↗γ'}
		&= \{ (q, ω, p, \push_j(γ'), q') ∈ δ ∣ γ ∈ p \}\,\text{, and} \\
	δ(q, q')|_{\down}^{γ↘∙}
		&= \{ (q, ω, p, \down, q') ∈ δ ∣ γ ∈ p \}\,\text{.}
\end{align*}

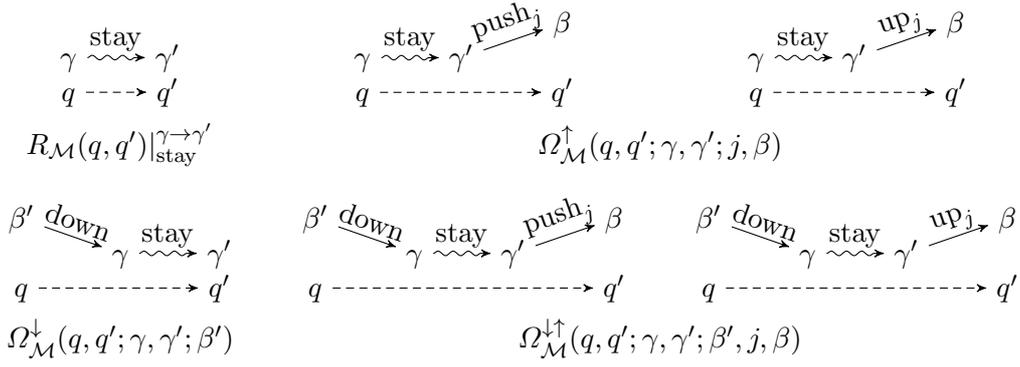
\begin{figure}[t]
\centering
\tikzset{ampersand replacement=\&, column sep=2em, inner ysep=0.1em, every node/.style={text height=.8em}}
\begin{tabular}{ccc}
	\begin{tikzpicture}
		\matrix (m) [matrix of math nodes] {
			γ \& γ' \\
			q \& q' \\
		};

		\draw [->, decorate] (m-1-1) to node [above, inner ysep=0.2em] {$\stay$} (m-1-2);
		\draw [->, densely dashed] (m-2-1) to (m-2-2);
	\end{tikzpicture}
	&
	\begin{tikzpicture}
		\matrix (m) [matrix of math nodes] {
			  \&    \& β \\
			γ \& γ' \&    \\
			q \&    \& q' \\
		};

		\draw [->, decorate] (m-2-1) to node [above, inner ysep=0.2em] {$\stay$} (m-2-2);
		\draw [->] (m-2-2) to node [above, sloped, inner sep=1pt] {$\push_j$} (m-1-3);
		\draw [->, densely dashed] (m-3-1) to (m-3-3);
	\end{tikzpicture}
	&
	\begin{tikzpicture}
		\matrix (m) [matrix of math nodes] {
				\&    \& β  \\
			γ \& γ' \&    \\
			q \&    \& q' \\
		};

		\draw [->, decorate] (m-2-1) to node [above, inner ysep=0.2em] {$\stay$} (m-2-2);
		\draw [->] (m-2-2) to node [above, sloped, inner sep=1pt] {$\up_j$} (m-1-3);
		\draw [->, densely dashed] (m-3-1) to (m-3-3);
	\end{tikzpicture}
	\\
	$R_ℳ(q, q')|_{\stay}^{γ → γ'}$
	&
	\multicolumn{2}{c}{$Ω_ℳ^{↑}(q, q'; γ, γ'; j, β)$}
	\\[.7em]
	\begin{tikzpicture}
		\matrix (m) [matrix of math nodes] {
			β' \&   \&    \\
			   \& γ \& γ' \\
			q  \&   \& q' \\
		};

		\draw [->] (m-1-1) to node [above, sloped, inner ysep=0.2em] {$\down$} (m-2-2);
		\draw [->, decorate] (m-2-2) to node [above, inner ysep=0.2em] {$\stay$} (m-2-3);
		\draw [->, densely dashed] (m-3-1) to (m-3-3);
	\end{tikzpicture}
	&
	\begin{tikzpicture}
		\matrix (m) [matrix of math nodes] {
			β' \&   \&    \& β  \\
			   \& γ \& γ' \&    \\
			q  \&   \&    \& q' \\
		};

		\draw [->] (m-1-1) to node [above, sloped, inner ysep=0.2em] {$\down$} (m-2-2);
		\draw [->, decorate] (m-2-2) to node [above, inner ysep=0.2em] {$\stay$} (m-2-3);
		\draw [->] (m-2-3) to node [above, sloped, inner sep=1pt] {$\push_j$} (m-1-4);
		\draw [->, densely dashed] (m-3-1) to (m-3-4);
	\end{tikzpicture}
	&
	\begin{tikzpicture}
		\matrix (m) [matrix of math nodes] {
			β' \&   \&    \& β  \\
			   \& γ \& γ' \&    \\
			q  \&   \&    \& q' \\
		};

		\draw [->] (m-1-1) to node [above, sloped, inner ysep=0.2em] {$\down$} (m-2-2);
		\draw [->, decorate] (m-2-2) to node [above, inner ysep=0.2em] {$\stay$} (m-2-3);
		\draw [->] (m-2-3) to node [above, sloped, inner sep=1pt] {$\up_j$} (m-1-4);
		\draw [->, densely dashed] (m-3-1) to (m-3-4);
	\end{tikzpicture}
	\\
	$Ω_ℳ^{↓}(q, q'; γ, γ'; β')$
	&
	\multicolumn{2}{c}{$Ω_ℳ^{↓↑}(q, q'; γ, γ'; β', j, β)$}
\end{tabular}
\caption{Groups of runs in $ℳ$ where dashed arrows signify the change of states and continuous arrows signify the change in the storage.}
\label{fig:puzzle-pieces}
\end{figure}

For every $q, q' ∈ Q$, $γ, γ' ∈ Γ ∪ \{@\}$, $β, β' ∈ Γ$, and $j ∈ ℕ_+$ we distinguish the following groups of runs (to help the intuition, they are visualised in \cref{fig:puzzle-pieces}):
\begin{enumerate}
	\item\label{item:up-run}
		A sequence of $\id$- or $\set$-instructions followed by an $\up$- or $\push$-instruction:
		\begin{align*}
			Ω_ℳ^{↑}(q, q'; γ, γ'; j, β)
				&= ⋃\nolimits_{\bar{q} ∈ Q} R_ℳ(q, \bar{q})|_{\stay}^{γ → γ'} ⋅ \big( δ(\bar{q}, q')|_{\push_j}^{γ' ↗ β} ∪ δ(\bar{q}, q')|_{\up_j}^{γ' ↗ ∙} \big)
		\end{align*}
	\item\label{item:down-run}
		A $\down$-instruction followed by $\id$- or $\set$-instructions:
		\[
			Ω_ℳ^{↓}(q, q'; γ, γ'; β')
				= ⋃\nolimits_{\bar{q} ∈ Q} δ(q, \bar{q})|_{\down}^{β' ↘ ∙} ⋅ R_ℳ(\bar{q}, q')|_{\stay}^{γ → γ'}
		\]
	\item\label{item:down-up-run}
		A $\down$-instruction, then a sequence of $\id$- or $\set$-instructions and finally an $\up$- or $\push$-instruction:
		\[
			Ω_ℳ^{↓↑}(q, q'; γ, γ'; β', j, β)
				= ⋃\nolimits_{\bar{q} ∈ Q} δ(q, \bar{q})|_{\down}^{β' ↘ ∙} ⋅ Ω_ℳ^{↑}(\bar{q}, q'; γ, γ'; j, β)
		\]
\end{enumerate}
The arguments of $Ω_ℳ^{↑}$, $Ω_ℳ^{↓}$, and $Ω_ℳ^{↓↑}$ are grouped using semicolons.
The first group describes the state behaviour of the run; the second group describes the storage behaviour at the parent position (i.e. the position of the $\set$- and $\id$-instructions), and the third group describes the storage behaviour at the child positions (i.e. the positions immediately above the parent position).

We build tuples of runs from the three groups above by matching the storage behaviour of neighbouring runs at the parent position.
A tuple $t = (θ_0, …, θ_{\ell})$ of runs is \emph{admissible}
	if $\ell = 0$ and $θ_0$ only uses $\id$- and $\set$-instructions; or
	if $\ell ≥ 1$, $θ_0$ is in group~\ref{item:up-run}, $θ_{\ell}$ is in group~\ref{item:down-run}, and for every $i ∈ [\ell]$, we have
\begin{align*}
	θ_{i-1} &∈ Ω_ℳ^{↑}(q, q'; γ, \bar{γ}; j, β) ∪ Ω_ℳ^{↓↑}(q, q'; γ, \bar{γ}; β', j, β)
	&&\text{and} \\*
	θ_i &∈ Ω_ℳ^{↓}(q'', q'''; \bar{γ}, γ'; β''') ∪ Ω_ℳ^{↓↑}(q'', q'''; \bar{γ}, γ'; β''', j', β'')
\end{align*}
for some $γ, \bar{γ}, γ' ∈ Γ ∪ \{@\}$, $β, β', β'', β''' ∈ Γ$,  $q, q', q'', q''' ∈ Q$, and $j, j' ∈ ℕ_+$.
Note that only the $\bar{γ}$ has to match.
Then $θ_{i-1}θ_i$ may \emph{not} be a run in $ℳ$ since it is not guaranteed that $q' = q''$ and $β = β'$.
We therefore say that there is a \emph{$(q', q'';j, β, β')$-gap between $θ_{i-1}$ and $θ_i$}.
Let $q_1, q_2 ∈ Q$ and $γ_1, γ_2 ∈ Γ ∪ \{@\}$.
We say that \emph{$t$ has type~$⟨q_1, q_2; γ_1, γ_2⟩$} if $\ell = 0$ and $θ_0 ∈ R_ℳ(q_1, q_2)|_{\stay}^{γ_1 → γ_2}$; or if $\ell ≥ 1$, the first transition in $θ_0$ has source state~$q_1$ and its predicate contains $γ_1$, the last transition of $θ_{\ell}$ has target state $q_2$, the last $\set$-instruction occurring in $t$, if there is one, is $\set(γ_2)$, and $γ_1 = γ_2$ if no $\set$-instruction occurs in $t$.
The set of admissible tuples in~$Ω_ℳ^*$ is denoted by $Ω_ℳ^{\star}$.
We define $t[y_1, …, y_{\ell}] = θ_0 y_1 θ_1 ⋯ y_{\ell} θ_{\ell}$ for every $y_1, …, y_{\ell} ∈ X$ to later fill the gaps with variables.

Let $T = (t_1, …, t_s) ∈ (Ω_ℳ^{\star})^*$ and $\ell_1, …, \ell_s$ be the counts of gaps in $t_1, …, t_s$, respectively.
For every $i ∈ [s]$ and $κ ∈ [\ell_i]$ we set $q_{(i,κ)}, q_{(i,κ)}' ∈ Q$, $β_{(i,κ)}, β_{(i,κ)}' ∈ Γ$, and $j_{(i,κ)} ∈ ℕ_+$ such that the $κ$-th gap in $t_i$ is a $(q_{(i,κ)}, q_{(i,κ)}'; j_{(i,κ)}, β_{(i,κ)}, β_{(i,κ)}')$-gap.
Let $φ_T, ψ_T: ℕ_+ × ℕ_+ → ℕ_+$ and $π_T: ℕ_+ × ℕ_+ → ℕ_+ × ℕ_+$ be partial functions such that for every $i ∈ [s]$ and $κ ∈ [\ell_i]$, the number $j_{(i,κ)}$ is the $φ_T(i, κ)$-th distinct number occurring in $J = j_{(1, 1)} ⋯ j_{(1, \lvert t_1 \rvert)} ⋯ j_{(s, 1)} ⋯ j_{(s, \lvert t_s \rvert)}$ when read left-to-right, $j_{(i,κ)}$ occurs for the $ψ_T(i, κ)$-th time at the element with index $(i, κ)$ in $J$, and $π_T(i, κ) = (φ_T(i, κ), ψ_T(i, κ))$.
Moreover let $m$ be the count of distinct numbers in $J$.
We call $T$ \emph{admissible} if
\begin{itemize}
	\item the $κ$-th run in $t_i$ ends with a $\push$-instruction whenever $φ_T(i, κ) = 1$,
	\item $β_{π_T^{-1}(κ', κ)}' = β_{π_T^{-1}(κ', κ+1)}$ for every $κ' ∈ [m]$ and $κ ∈ [\ell_{κ'}-1]$, and
	\item there are $q_1, \bar{q}_1, …, q_s, \bar{q}_s ∈ Q$ and $γ_0, …, γ_s ∈ Γ ∪ \{@\}$ such that for every $κ ∈ [s]$, we have that $t_κ$ is of type~$⟨q_κ, \bar{q}_κ; γ_{κ-1}, γ_κ⟩$.
\end{itemize}
We then say that \emph{$T$ has~$(A; B_1, …, B_m)$}, denoted by $\type(T) = (A; B_1, …, B_m)$, where $A = ⟨q_1, \bar{q}_1, …, q_s, \bar{q}_s; γ_0, …, γ_s⟩$, and for every $κ' ∈ [m]$:
\begin{align*}
	B_{κ'} &= ⟨q_{π_T^{-1}(κ',1)}, q_{π_T^{-1}(κ',1)}', …, q_{π_T^{-1}(κ',\ell_{κ'})}, q_{π_T^{-1}(κ',\ell_{κ'})}'; \\*
				&\qquadβ_{π_T^{-1}(κ',1)}, β_{π_T^{-1}(κ',1)}', …, β_{π_T^{-1}(κ', \ell_{κ'})}'⟩\,\text{.}
\end{align*}
The set of admissible elements of $(Ω_ℳ^{\star})^*$ is denoted by $Ω_ℳ^{\star\star}$.

\begin{construction}\label{con:automaton_to_MCFG}
	Let $ℳ = (Q, \TS{Γ}, Σ, q_{\text{i}}, \{\underline{(ε, @)}\}, δ, Q_{\text{f}})$ be a cycle-free $k$-restricted TSA in stack normal form.
	Define the $k$-MCFG $G'(ℳ) = (N, Σ, I, R')$ where
		$N = \{A, B_1, …, B_m ∣ ⟨A; B_1, …, B_m⟩ ∈ \type(Ω_ℳ^{\star\star}) \}$,
		$I = \{ ⟨q_{\text{i}}, q; @, @⟩ ∣ q ∈ Q_{\text{f}} \}$, and
		$R'$ contains for every $T = (t_1, …, t_s) ∈ Ω_ℳ^{\star\star}$ the rule
			\( A → \big[ u_1, …, u_s \big](B_1, …, B_m)\)
			where $(A; B_1, …, B_m)$ is the type of $T$ and
			$u_κ = t_κ[x_{φ_T(κ,1)}^{ψ_T(κ,1)}, …, x_{φ_T(κ,\ell_κ)}^{ψ_T(κ,\ell_κ)}]$ for every $κ ∈ [s]$.
	Let $G(ℳ)$ be a $k$-MCFG recognising $\{⟦θ⟧ ∣ θ ∈ L(G'(ℳ))\}$.\footnote{The $k$-MCFG $G(ℳ)$ exists since $⟦⋅⟧$ is a homomorphism and $k$-MCFLs are closed under homomorphisms \cite[Thm.~3.9]{SekMatFujKas91}.}
\end{construction}

\begin{proposition}\label{lem:automaton_equals_MCFL}
	$L(ℳ) = L(G(ℳ))$ for every cycle-free $k$-restricted TSA $ℳ$ in stack normal form.
\end{proposition}
\begin{proof}
	We can show by induction that $G'(ℳ)$ generates exactly the valid runs of~$ℳ$.
	Our claim then follows from the definition of language of $G(ℳ)$.
\end{proof}

\subsection{The main theorem}

\begin{theorem}\label{thm:equivalence-MCFL-rTSL}
	Let $L ⊆ Σ^*$ and $k ∈ ℕ_+$.
	The following are equivalent:
	\begin{enumerate}
		\item\label{item:thm:kMCFG} There is a $k$-MCFG $G$ with $L = L(G)$.
		\item\label{item:thm:kAutomaton} There is a $k$-restricted tree stack automaton $ℳ$ with $L = L(ℳ)$.
	\end{enumerate}
\end{theorem}
\begin{proof}
	We get the implication (\ref{item:thm:kMCFG} $⟹$ \ref{item:thm:kAutomaton}) from \cref{lem:MCFL_Automaton_restricted,prop:PMCFL_toAutomaton} and the implication (\ref{item:thm:kAutomaton} $⟹$ \ref{item:thm:kMCFG}) from \cref{lem:cycle-free,lem:stack-normal-form,lem:automaton_equals_MCFL}.
\end{proof}


\section{Conclusion}

The automata characterisation of multiple context-free languages presented in this paper is achieved through tree stack automata that possess, in addition to the usual finite state control, the ability to manipulate a tree-shaped stack; tree stack automata are then restricted by bounding the number of times that the stack pointer enters any position of the stack from below (cf. \cref{sec:tree_stack_operator}).
The proofs for the inclusions of multiple context-free languages in restricted tree stack languages and vice versa are both constructive; the former even works for parallel multiple context-free grammars, although the resulting automaton may then no longer be restricted (cf. \cref{sec:automata_equal_MCFG}).
\Cref{thm:equivalence-MCFL-rTSL} closes a gap in formal language theory open since the introduction of MCFGs \cite{SekMatFujKas91}. 
The proof allows for the easy implementation of a parser for parallel multiple context-free grammars.

\bibliographystyle{alpha-doi}
\bibliography{references}

\end{document}